\documentclass{article}

\usepackage{amsmath}
\usepackage{float}
\usepackage{amssymb}
\usepackage[ruled,vlined]{algorithm2e}

\usepackage[authoryear]{natbib}
\usepackage{graphicx}
\usepackage{subfigure}
\usepackage{algorithm2e}
\usepackage{amsthm}
\usepackage{enumerate}

\newtheorem{theorem}{Theorem}

\newtheorem{lemma}{Lemma}

\newtheorem{corollary}{Corollary}

\newtheorem{definition}{Definition}

\makeatletter
\renewcommand{\algocf@captiontext}[2]{#1\algocf@typo. \AlCapFnt{}#2} 
\def\@algocf@capt@plain{top}
\renewcommand{\algocf@makecaption}[2]{%
	\addtolength{\hsize}{\algomargin}%
	\sbox\@tempboxa{\algocf@captiontext{#1}{#2}}%
	\ifdim\wd\@tempboxa >\hsize
	\hskip .5\algomargin%
	\parbox[t]{\hsize}{\algocf@captiontext{#1}{#2}}
	\else%
	\global\@minipagefalse%
	\hbox to\hsize{\box\@tempboxa}
	\fi%
	\addtolength{\hsize}{-\algomargin}%
}
\makeatother

\makeatletter
\newcommand{\disjoint}{\mathrel{\mathpalette\disj@int\relax}}
\newcommand{\disj@int}[2]{%
	\ooalign{%
		\hfil$\m@th#1\cap$\hfil\cr
		\hfil$\m@th#1\diagup$\hfil\cr
	}%
}
\makeatother

\DeclareMathOperator{\pr}{pr}
\DeclareMathOperator{\cor}{cor}

\DeclareMathOperator{\Min}{Min}

\usepackage{url}
\usepackage{xr}
\begin{document}
	
	%
	%
	
	\title{Hypergraph Variable Selection with False Discovery Rate Control}
	
	\author{Sarah Organ, Toby Kenney, Hong Gu \\
	Department of Mathematics and Statistics, Dalhousie University.}

        \maketitle

\begin{abstract}
  Variable selection methods that control the false discovery rate
  often lose power when predictors exhibit complex dependence
  structures. We previously showed that selecting hierarchically
  clustered groups of predictors can mitigate this issue while
  maintaining false discovery rate control. When correlations are less
  structured, however, overlapping predictor sets may be more
  effective. We introduce a generalized false discovery rate for
  hypotheses defined on sets of predictors and propose a
  hypergraph-based selection method. This approach achieves higher
  power across diverse settings while preserving rigorous false
  discovery rate control.
\end{abstract}

\subsubsection*{keywords}
false discovery rate; variable selection; hypergraph methods; multiple hypothesis testing; dependence structure; linear step-up procedures.

\section{Introduction}\label{hvsintro}

Variable selection with false discovery rate (FDR) control is an
important problem in statistics. Classical approaches test, for
each predictor $X_i$, the null hypothesis that it is non-informative,
then apply multiple-testing procedures to guarantee the expected
proportion of false discoveries remains below a target level. Widely
used procedures such as BH~\citep{BH1995} and BY~\citep{BY2001}, apply
broadly provided valid $p$-values are available.

Alternative methods --- such as model-X knockoffs~\citep{knockoff3} and
Gaussian mirror procedures~\citep{GaussianMirrors2023,DS} --- replace
$p$-values with model-dependent test statistics. These approaches also
control FDR through data-adaptive thresholds but are typically
restricted to generalized linear model settings.

However, 
both types of methods underperform when predictors are highly
correlated. In this case, statistical evidence may be insufficient to
distinguish true signals from surrogate variables. Many
FDR-controlling methods respond by selecting neither variable,
resulting in reduced power and sometimes inflated
FDR~\citep{correlation,DS}.

To address this, recent work by the authors~\citep{SHRED} introduces a
setwise variable selection paradigm. When the data
cannot identify the true variable from among highly correlated predictors,
the methods select those candidate variables as a set, deemed to
contain at least one true predictor. 
Because of implications between hypotheses, we needed to
develop substantial generalizations of the linear step-up methods for
multiple testing. These generalizations incorporate the weighted BH
method~\citep{wBH} and also the ``outer nodes
FDR''~\citep{Yekutieli2006}. These setwise variable selection methods
maintain FDR control across a wide range of simulation studies while
achieving higher power than existing methods.

Although setwise selection increases power, the methods developed
by~\cite{SHRED} only test sets of variables from a hierarchical
clustering, which can be overly restrictive. In spatial, temporal, or
other structured 
applications, groups of correlated predictors often overlap, and thus
do not fall into a hierarchical structure. As a result, existing
cluster-based methods may fail to test the most relevant variable
sets, reducing power. This motivates the question addressed in this
paper: how should we conduct setwise variable selection when the
tested sets overlap?

While there has been substantial work on hierarchical multiple
hypothesis testing, for example~\citep{Yekutieli2008,TreeBH2021},
research on FDR control in non-hierarchical structures is very
limited. The DAGGER method~\citep{DAGGER2019} can apply to a general
partially-ordered set of hypotheses. However, the definition of FDR
controlled by the procdure is not appropriately adjusted to account
for the partial order structure on the hypotheses. Focused BH and
variants~\citep{FocusedBH2023,WFBH2025} control FDR in cases where
only a subset of rejected hypotheses are counted as discoveries ---
for example, in hierarchical clustering situations, only the ``outer
nodes'' would be counted as discoveries. The same outer node filter
could be applied to any partial order structure on the hypotheses, but
the appropriate cut-offs for controlling FDR have not been developed
in general, so would need to be determined for each specific
structure. Furthermore, the outer node filter on a partially-ordered
set does not take into account the additional structure for families
of sets of variables. Other methods that can apply in non-hierarchical
situations, such as the $p$-value filter~\citep{Ramdas2019}
deal with different problems, such as simultaneous control of multiple
FDRs.

There has been more work on non-hierarchical multiple testing in the
familywise error rate (FWER) control literature, for example closed
testing methods~\citep{Marcus,Hommel1988} can in principle accommodate
arbitrary sets of hypotheses, but they require testing all subsets,
which is computationally infeasible even for moderate dimensions. The
inheritance principle~\citep{GoemanFinos2012} was developed for
hierarchical-structured hypotheses, but could be extended to
non-nested hypotheses. However, this would require developing custom
inheritance schemes for each application.

Modern methods which eschew $p$-values also have a few developments
for structured hypotheses, such as Multi-resolution knockoff
procedures~\citep{Katsevich2019,Sesia2020}, which  control FDR
simultaneously across layers of hypotheses, similar to the $p$-value
filter. These could be overlapping layers, but each layer consists
of non-overlapping sets. They do not apply to the overlapping-set setting considered here.

The general linear step-up procedure of~\citet{SHRED} can operate on
arbitrary families of hypotheses, but relies on a \emph{sizing
function}~$\sigma$ that determines how many discoveries to count when
a set of hypotheses is rejected. Developing such a sizing function is
the main obstacle for overlapping predictor sets.  In this paper, we
introduce a framework for setwise variable selection using
\emph{hypergraphs}, which naturally represent overlapping sets. We
develop appropriate discovery-counting mechanisms and provide
corresponding FDR guarantees.


\section{Variable Selection on a Hypergraph}\label{hvsvarselection}

We let $X=\{X_1,\ldots,X_p\}$ denote the set of predictor variables,
and $Y$ the response variable. We let $T$ denote the set of true
predictors, and we assume that for any subset $C\subseteq X$, we have
a valid hypothesis test (e.g. a likelihood ratio test) for the null
hypothesis $H_C: C\cap T=\emptyset$. We use $\mathcal{P}(X) = \{A|A
\subseteq X\}$ to denote the powerset of $X$. We recall the
definition of a hypergraph:


\begin{definition}
  A {\em hypergraph} $(X,E)$ consists of a set $X$ of {\em vertices},
  and a set $E\subseteq {\mathcal {P}(X)}$ of {\em edges}, where
  ${\mathcal {P}(X)}$ is the powerset of $X$ (the collection of all
  subsets of $X$).
\end{definition}

We assume that there is at most one edge for any subset
of vertices, and we allow edges for singleton
subsets, but not the empty subset. In our problem, $X$ is the set
of predictor variables, and $E$ is initially the collection,
$\mathcal S$, of sets of variables for which we have tested
the corresponding hypotheses. Without loss of generality, we
assume $X = \bigcup \mathcal{S}$. After we reject a set ${\mathcal
  R}\subseteq{\mathcal S}$ of hypotheses, we consider the
subhypergraph $(X,{\mathcal R})$ of sets corresponding to rejected
hypotheses.


In~\citep{SHRED}, we argued that simply counting the number of
rejected hypotheses is not an appropriate way to assess either the FDR
or power.  There are two main reasons for this. Firstly, some null
hypotheses represent stronger assertions than others; rejecting a
stronger hypothesis is easier, and therefore such rejections should
carry less weight. Secondly, when tested sets have non-empty
intersections, the corresponding hypotheses are dependent. A single
true variable can make several null hypotheses false simultaneously,
so rejecting one null reduces the informational gain from rejecting
others.  To address these issues, we introduced a generalized notion
of FDR and power based on a sizing function $\sigma$, which maps sets
of rejected hypotheses to $\mathbb{R}_{\ge 0}$, rather than relying on
set cardinality. The idea is that $\sigma(A)$ represents the ``number
of discoveries'' from rejecting all hypotheses in the set $A$.

The setting of~\cite{SHRED} was simplified by the hierarchical tree
structure. For any two tested sets, either they are disjoint or one
contains the other. For this structure, $\sigma$ can be set as the sum
of weights of the minimal rejected hypotheses --- that is, those
corresponding to sets for which no proper subset was also
rejected. However, this sizing function is not appropriate for a more
general hypergraph. For example, if
$\mathcal{R}=\{\{1,2\},\{1,3\},\ldots,\{1,p\}\}$, then all $p-1$
elements of $\mathcal{R}$ are minimal, so we would add all their
weights, but rejecting all these null hypotheses is still weaker than
rejecting the single null hypothesis $H_{\{1\}}$, so the sum of all
$p-1$ minimal weights would need to be smaller than 1. As $p$ grows
large, this would only be possible if the weight of all non-singleton
hypotheses goes to 0. In this paper, we develop a new sizing function
for this broader setting.

The subsets of $X$ are ordered by containment, and the corresponding
null hypotheses are ordered by reverse implication. That is, if
$A\subseteq B\subseteq X$, then $H_A\Leftarrow H_B$. Because of this
implication, if we reject $H_A$, then we implicitly reject $H_B$.
Let $\overline{\mathcal R}$ denote the set of all implicitly and
explicitly rejected hypotheses, i.e. the set of all
subsets of $X$ that contain an edge of the hypergraph $G=(X,{\mathcal
  R})$. We define the sizing function $\sigma({\mathcal R})$ as a
function of the cardinality $\big|\overline{\mathcal R}\big|$. Since
rejecting weaker null hypotheses leads to more implicit rejections
than rejecting stronger null hypotheses, this gives more weight to
rejecting weaker null hypotheses, and also accounts for the dependence
between hypotheses. In fact, sets that are not in
$\overline{\mathcal R}$ are well-studied in (hyper)graph theory. They
are called independent sets. (This is not related to the notion of
independence in statistics.)

\begin{definition}
  For a hypergraph $G=(X,{\mathcal R})$, a vertex set, $I
  \subseteq X$ is an {\em independent set} of $G$ if no hyperedge $R
  \in {\mathcal R}$, is a subset of $I$. We use $IS(G)$ to denote the
  set of all independent sets of $G$.
\end{definition}



It is more convenient to describe the number of
discoveries as a function of $|IS(G)|=2^{|X|}-\big|\overline{\mathcal
  R}\big|$. The motivating idea is that selecting
one predictor should count as one discovery. If we select predictor
$X_i\in X$, i.e. we reject the null hypothesis corresponding to the
set $\{X_i\}$, then we implicitly reject the hypotheses corresponding
to all subsets containing $\{X_i\}$. If we have not previously
explicitly rejected hypotheses corresponding to any sets containing
$X_i$, then by rejecting $H_{\{X_i\}}$, we are implicitly rejecting
exactly half of the remaining hypotheses (independent sets). That is,
each discovery should halve the number of independent sets. This gives
us the sizing function
\begin{equation*}
	\sigma({\mathcal R})=|X| - \log_2(|IS(G)|)
\end{equation*}
which reflects the information gained from explicitly and
implicitly rejected subsets of predictors.

In the classical setting of single-variable hypotheses, let $R =
\bigcup \mathcal{R}$ denote the set of rejected singletons. In this
case, the family of independent sets is given by
\[
    IS(X, \mathcal{R}) = \{\, A \subseteq X \setminus R \,\},
\]
from which it follows that
\[
    |IS(G)| = 2^{|X \setminus R|} = 2^{\,|X| - |R|}.
\]
Consequently, the effective rejection size satisfies
\[
    |X| - \log_2 |IS(G)|
    = |X| - \log_2\!\left( 2^{\,|X| - |R|} \right)
    = |R|.
\]
Thus, in this classical regime, the proposed measure reduces exactly to the 
cardinality of the set of selected variables, so this is a
generalization of the classical case.

The complement of an independent set is a vertex cover:

\begin{definition}
  For a hypergraph $G=(X,{\mathcal R})$, a {\em vertex cover} of $G$
  is a subset of vertices, $C \subseteq X$, such that for every
  hyperedge $R \in {\mathcal R}$, at least one node in $C$ is in $R$,
  i.e. for all $R \in \mathcal{R}, C \cap R \neq \emptyset$.
\end{definition}

A vertex cover $C$ is a collection of predictor variables that could
be the true variables without any of the rejected hypotheses being
false positives. From the predictive modelling point of view, any
vertex cover $C$ should be a good set of predictors.

This gives rise to an information-theoretic interpretation of our
sizing function. Initially, there are $2^{|X|}$ possibilities for
which of the predictors are true predictors. Thus, from an information
theoretic point of view, it requires $|X|$ bits of information to
identify the true predictors. If we reject a set $\mathcal R$ of
hypotheses, then the number of possibilities for the set of true
predictors is the number of vertex covers of $G=(X,{\mathcal
  R})$. Thus, it requires $\log_2(|IS(G)|)$ bits to identify the true
predictors. Thus, the information gained from rejecting these
hypotheses is $|X|-\log_2(|IS(G)|)=\sigma(\mathcal{R})$.

In the setting considered by \citep{SHRED}, where the tested variable
sets arise from hierarchical clustering, the
sizing function can be expressed as a sum of weights associated with
the minimal rejected sets, analogous to the sizing functions examined
in that work. Let $S_1, \ldots, S_m$ be the minimal rejected
sets. Because these sets are disjoint, an independent set of the
hypergraph $(X, \{S_1, \ldots, S_m\})$ must include a proper (possibly
empty) subset of each $S_i$, together with an arbitrary subset of $X
\setminus (S_1 \cup \cdots \cup S_m)$.  Thus the total number of
independent sets is
\[
|IS(G)| = 2^{\,|X \setminus (S_1 \cup \cdots \cup S_m)|} \prod_{i=1}^m \left( 2^{|S_i|} - 1 \right)
\]
where $2^{|S_i|} - 1$ is the number of proper subsets of $S_i$. It follows that the sizing function is
\[
\sigma({\mathcal R}) = \mspace{-16mu}\sum_{S \in \Min({\mathcal R})}\mspace{-16mu} w_S, \qquad 
w_S = \log_2\!\left( \frac{2^{|S|}}{\,2^{|S|}-1\,} \right)
\]
where $\Min({\mathcal R})$ denotes the set of minimal elements of the
partially ordered set $\mathcal R$. Thus, the theoretical framework
developed in \citep{SHRED} extends naturally to this sizing function,
although the resulting weights differ from the choice $w_S = 1/|S|$
used in~\cite{SHRED}.


\section{Sizing Function, Power and FDR, and the Procedure}\label{hvsproc}

\subsection{Notation and Generalized Linear Step-up Procedure}

Throughout this article, let $\mathcal{S} = \{S_1, \ldots, S_m\}$ denote the
collection of tested subsets, and assume without loss of generality that
$X = \bigcup \mathcal{S}$. For each $i \in \{1,\ldots,m\}$, let $H_i$
denote the hypothesis corresponding to subset $S_i$.

Let $X_A$ denote the set of true variables and $X_N = X \setminus X_A$
the set of noise variables. Define $I_N = \{\, i \in \{1,\ldots,m\} : S_i \cap X_A = \varnothing \,\}$
as the index set of true null hypotheses --- subsets that
contain no true variables. Its complement, $I_A = \{1,\ldots,m\} \setminus I_N$ 
is the index set of false null hypotheses --- subsets containing at
least one true variable. Let $\mathcal{P} = \{\, J \subseteq \{1,\ldots,m\} \,\}$
denote the powerset of all hypothesis indices, representing all possible
selections of hypotheses.
We use the sizing
function 
$\sigma({\mathcal R})=|X|-\log_2(|IS(G(X,\mathcal{R}))|)$ from
Section~\ref{hvsvarselection}, to define the generalized power
(gPower) and generalized FDR (gFDR) for a generalized variable
selection method.
\begin{equation*}
	\textrm{gFDR} = \mathbb{E}\left(\frac{\sigma(\mathcal{R} \cap I_N)}{\sigma(\mathcal{R})}\right) =
        \mathbb{E}\left(\frac{|X|-\log_2(|IS(G(X, \mathcal{R} \cap
          I_N))|)}{|X|-\log_2(|IS(G(X,\mathcal{R}))|)}\right)
\end{equation*}
\begin{equation*}
	\textrm{gPower} =
        \mathbb{E}\left(\frac{\sigma(\mathcal{R})-\sigma(\mathcal{R}
          \cap I_N)}{|X_A|}\right) = \mathbb{E}\left(\frac{\log_2(|IS(G(X, \mathcal{R} \cap I_N))|)-\log_2(|IS(G(X,\mathcal{R}))|)}{|X_A|}\right)
\end{equation*}
When $\mathcal S$ is the set of singletons we have $\sigma(A)=|A|$ for
every $A\subseteq {\mathcal S}$, which gives the usual definitions,
$\textrm{gFDR} = \mathbb{E}\left(\frac{\sigma(\mathcal{R} \cap
  I_N)}{\sigma(\mathcal{R})}\right) = \mathbb{E}
\left(\frac{|\mathcal{R} \cap I_N|}{|\mathcal{R}|}\right) = \textrm{FDR}$ and
$\textrm{gPower} = \mathbb{E}\left(\frac{\sigma(\mathcal{R})-\sigma(\mathcal{R} \cap
  I_N)}{|X_A|}\right) = \mathbb{E}\left(\frac{|\mathcal{R}|-|\mathcal{R} \cap
  I_N|)}{|X_A|}\right) = \mathbb{E}\left(\frac{|\mathcal{R} \cap
  I_A|}{|X_A|}\right) = \textrm{power}$.

The reader may wonder
why the number of false
discoveries is $\sigma(\mathcal{R}\cap I_N)$ and the number of true
discoveries is $\sigma(\mathcal{R})-\sigma(\mathcal{R}\cap I_N)$.
The subadditivity $\sigma(\mathcal{R}\cap I_N)+\sigma(\mathcal{R}\cap
I_A)-\sigma(\mathcal{R})$ comes from an overlap between the content of
the falsely rejected and correctly rejected hypotheses. The content of
this overlap should all be considered falsely rejected.  For example,
let $\mathcal{R}=\{\{X_1,X_2\},\{X_1,X_3\}\}$. The null hypotheses are
$H_{\{X_1,X_2\}}: \{X_1,X_2\} \cap X_A = \emptyset $ and
$H_{\{X_1,X_3\}}:\{X_1,X_3\} \cap X_A = \emptyset $. The shared
content between these hypotheses is $H_{\{X_1\}}: \{X_1\} \cap X_A =
\emptyset$. If, for example, $H_{\{X_1,X_2\}}$ is in $I_N$ and
$H_{\{X_1,X_3\}}$ is in $I_A$, then the shared content $H_{\{X_1\}}$
is a true null, so its contribution to the sizing function should be
considered as a false discovery. From the information-theoretic point
of view, $\sigma({\mathcal R})-\sigma(\mathcal{R}\cap
I_N)=\log_2(|IS(G_N)|)-\log_2(|IS(G)|)$, where $G_N=(X,\mathcal{R}\cap
I_N)$. This is the negative base-2 logarithm of the probability that a
random independent set for $G_N$ is independent for $G$. That is, it
is the amount of additional information given by rejecting ${\mathcal
  R}\cap I_A$, when we already know that the hypotheses in ${\mathcal
  R}\cap I_N$ are true nulls. Note that because of the subadditivity,
this is the strictest possible way to count the true positive rate.

We now recall the generalized linear step-up procedure (GLSUP) from \citep{SHRED}:
\begin{definition}[Generalized Linear Step-up Procedure]
  Let $H_1,\ldots,H_m$ be a set of null hypotheses with corresponding
  conservative $p$-values $p_1,\ldots,p_m$ and sizing function
  $\sigma:\mathcal{P} \rightarrow \mathbb R_{\geqslant 0}$ where
  $\mathcal{P}$ is the powerset of $\{1,\ldots,m\}$. For any $c \in
  [0,1]$, define $I_c = \{i\in \{1,\ldots,m\}|p_i \leqslant c\}$. We
  reject all hypotheses with indices in
  $\overline{I_{c_{\textrm{max}}}}$ where $c_{\textrm{max}} =
  \sup\left\{c \in [0,1]\middle| \sigma\left(\overline{I_c}\right)
  \geqslant \alpha c\right\}$, where $\overline{I_c}$ is the set of
  all hypotheses that imply hypotheses in $I_c$.
\end{definition}

The slope $\alpha$ is chosen to control the gFDR at the chosen level
$q$ under particular assumptions. For example, for the BH method,
$\alpha=\frac{m}{q}$, and for the BY method,
$\alpha=\frac{m\sum_{j=1}^{m}\frac{1}{j}}{q}$. In \citep{SHRED},
we proved gFDR control under certain assumptions for general choices
of the sizing function $\sigma$. In Section \ref{hvsfdr}, we apply
these results to our hypergraph variable selection method.

\subsection{The Hypergraph Variable Selection (HVS) Procedure}
We now provide the full HVS procedure:
\begin{description}

    \item[Inputs:] The predictor variable set $X$, the response variable $Y$, the data 
    consisting of observations of $(X, Y)$, and the collection 
    $\mathcal{S} \subseteq \mathcal{P}(X)$ of predictor subsets to be tested.

    \item[1.] Fit the full model predicting $Y$ from all variables in $X$.
    
    \item[2.] For each set $S\in{\mathcal S}$, fit the model with the
	variables in $X \backslash S$ and compare to the full model using
	the likelihood ratio test (or another appropriate test) to obtain
	the $p$-values.

      \item[3.] Apply the GLSUP with cut-off $\alpha_{HVS}$ (defined in
	Section~\ref{hvsfdr}) and sizing function $\sigma({\mathcal R}) =
	|X|-\log_2(|IS(G_{\mathcal R})|)$, where $G_{\mathcal
		R}=(X,{\mathcal R})$ is the hypergraph with vertex set $X$, and
	edge set $\mathcal R$.
 \end{description} 

This procedure yields a hypergraph of selected predictor
variables. Any minimal vertex cover of this hypergraph constitutes an
effective set of predictors for modelling the response.  The threshold
$\alpha_{\mathrm{HVS}}$ is chosen to control the gFDR at level
$q$. The value $\alpha_{\mathrm{HVS}}$ depends on the collection
$\mathcal{S}$ of subsets tested, and on the assumptions made about
null $p$-values. In Section~\ref{hvsfdr}, we provide appropriate
values of $\alpha_{\mathrm{HVS}}$ for the case in which all subsets of
size at most two are tested.


\subsection{Counting the Number of Independent Sets}

The sizing function $\sigma$ depends on the number of independent sets
of the associated hypergraph. Thus, in order to apply the GLSUP, we
need to count the number of independent sets in a given
hypergraph. Computing this quantity is NP-hard, so exact algorithms
are not computationally feasible for realistic applications of our
method. We therefore propose the importance-sampling based method in
Algorithm~\ref{CountIS} to approximately count the number of
independent sets. This algorithm is a special case of the general
backtracking approximation algorithm of~\cite{knuth1975}.

\setlength{\algomargin}{0pt} 
\setcounter{AlgoLine}{0}
\begin{algorithm}[htbp]
	\DontPrintSemicolon
	\SetAlgoNlRelativeSize{0} 
	\SetAlgoNlRelativeSize{0}
	\SetAlgoNlRelativeSize{0}
	\setlength{\textwidth}{500pt} 
	\renewcommand{\baselinestretch}{1.2}
	\SetAlgoLined
	\KwIn{Hypergraph \(G = (V, E)\), number of sequences to subsample \(N\)}
	\KwOut{Approximate count of independent sets}
	
	\SetKwFunction{SampleIndependentSequence}{SampleIndependentSequence}
	\SetKwFunction{ComputeWeight}{ComputeWeight}
	
	Initialize \( \text{totalWeight} \leftarrow 0 \)\;
	
	\For{\(i \leftarrow 1\) \KwTo \(N\)}{
		\(w \leftarrow \SampleIndependentSequence(G)\)
                \tcp*[r]{Sample a sequence of independent sets
                  \(S_0\subseteq S_1\subseteq\cdots \subseteq V\)}
		
		\(\text{totalWeight} \leftarrow \text{totalWeight} + w\)\;
	}
	\Return \(\text{estimate} \leftarrow \frac{\text{totalWeight}}{N}\)\;
	\vspace{1em}
	
	\SetKwProg{Fn}{Function}{:}{}
	\Fn{\SampleIndependentSequence{\(G\)}}{
		Initialize \(G' \leftarrow G\), \(w \leftarrow 1\), \(\textrm{total} \leftarrow 1\), \(i \leftarrow 1\)\;
		Remove all vertices with loops (hyperedges with only one vertex) in $G'$\;
		\While{\(V(G') \neq \emptyset\)}{
			Set $w \leftarrow w \times |V(G')|/i$ \;
			Set $i \leftarrow i+1$\;
			Set $\textrm{total}\leftarrow \textrm{total}+w$ \;
			Select a random vertex $v \in V(G')$ and remove $v$ from $V(G')$ and from all elements of $E(G')$ that contain $v$; \;
			Remove all vertices with loops (hyperedges with only one vertex) in $G'$\;
		}
		\Return \(\textrm{total}\)\;
	}	
	\caption{Counting the Number of Independent Sets via Subsampling}\label{CountIS}
\end{algorithm}

Algorithm~\ref{CountIS} uses sequential importance sampling. The
\texttt{SampleIndependentSequence} function samples a sequence of
subsets from a sampling distribution with support only on the
independent sets, where each step in the sequence corresponds to a
single sample. At each step, a vertex is added to the current set,
chosen randomly from the set of vertices that can be added while
retaining independence. The weight $w$ for each sample is calculated
sequentially based on the number of possible choices at each
stage. The hypergraph $G'$ keeps track of which vertices can still be
added at each step. Details of how this works are in
Appendix~\ref{SUPPCountIndepProof}. In practice, we set the number of
sequences to sample to $N = 1000$, which yields sufficiently accurate
approximations while keeping computation tractable, even for
hypergraphs of moderate to large size.

To implement the GLSUP, we can use this method for every cut-off, or
at least at those corresponding to $p$-values of tested hyperedges.
However, we can adapt the algorithm to reduce computation when
approximately counting independent sets across a nested sequence of
hypergraphs defined on the same vertex set.  Consider a sequence
${\mathbb G}=G_0,G_1,\ldots,G_{n_e}$, where $G_i=(V,E_i)$, and the
hyperedge sets form a nested sequence $\emptyset=E_0\subseteq
E_1\subseteq \cdots\subseteq E_{n_e}$ ordered by increasing
$p$-values. For any $i\geqslant j$, every independent set of $G_i$ is
also an independent set of $G_j$.  This allows us to reuse sequences
generated for $G_j$ to estimate $|IS(G_i)|$ when $i\geqslant
j$. Specifically, when using samples from $G_j$ to estimate the number
of independent sets in $G_i$, we assign zero weight to those sampled
sets that are not independent in $G_i$. This still yields an
unbiased estimator for $|IS(G_i)|$.  However, although these reused
samples provide unbiased estimates, they have higher variance than
estimates obtained from samples generated using $G_i$.  We therefore
down-weight the contribution of sequences generated from earlier $G_j$
and use them to partially replace samples that would otherwise be
drawn from $G_i$, thereby reducing the total number of sequences
generated.  The down-weighting factor is given by $\prod_{e_k\in
  E_i\setminus E_j}(1-2^{-|e_k|})$, which is a lower bound on the
proportion of independent sets of $G_j$ that remain independent in
$G_i$. With this adjustment, the total number of sequences of
independent sets of $G_i$ included among our samples by
Algorithm~\ref{CountISModified} should still exceed the desired
threshold $N$.  We implement this in the \texttt{R} package
\texttt{hypergraph.sizing} to create a sizing function for use with
the \texttt{GLSUP} package.

\setlength{\algomargin}{0pt} 
\setcounter{AlgoLine}{0}
\begin{algorithm}[htbp]
	\DontPrintSemicolon
	\SetAlgoNlRelativeSize{0} 
	\SetAlgoNlRelativeSize{0}
	\SetAlgoNlRelativeSize{0}
	\setlength{\textwidth}{500pt} 
	\renewcommand{\baselinestretch}{1.2}
	\SetAlgoLined
	\KwIn{Nested sequence of hypergraphs \({\mathbb G}=(G_i = (V,
          E_i))_{i=0}^{n_e}\), where $E_i=\{e_1,\ldots,e_i\}$, number
          of sequences to subsample \(N\)}
	\KwOut{Approximate counts of independent sets for each $G_i$}
	
	\SetKwFunction{SampleIndependentSequence}{SampleIndependentSequence}
	\SetKwFunction{ComputeWeight}{ComputeWeight}
	
	Initialize \( \textbf{totalWeight} \leftarrow \mathbf{0},\;\textbf{sampleSize} \leftarrow \mathbf{0}, \)\;        
        
	\For{\(i \leftarrow 1\) \KwTo \(n_e\)}{
          \textrm{sampleSize}$_{i}\leftarrow$\textrm{sampleSize}$_{i-1}$\(*(1-2^{-|e_i|})\)\;
          \For{\(j\leftarrow i\) \KwTo \(n_e\)}{
            \textrm{totalWeight}$_{j}\leftarrow$\textrm{totalWeight}$_{j}$\(*(1-2^{-|e_i|})\)\;
          }
	  \While{\textrm{sampleSize}$_{i}$ \(<N\)}{
		\(\mathbf{w} \leftarrow
                \SampleIndependentSequence({\mathbb G},i)\)
                \tcp*[r]{Sample a sequence of independent sets \(S
                  \subseteq V\) for $G_i$, and use them to estimate
                  $|IS(G_j)|$ for all $j\geqslant i$.}
		
		\(\textbf{totalWeight} \leftarrow \textbf{totalWeight}
                + \mathbf{w}\)\;
                \(\textrm{sampleSize}_{i}\leftarrow \textrm{sampleSize}_{i}+1\)\;
	  }
        }
	\Return \(\text{estimate} \leftarrow \frac{\textbf{totalWeight}}{\textbf{sampleSize}}\)\;
	\vspace{1em}
	
	\SetKwProg{Fn}{Function}{:}{}
	\Fn{\SampleIndependentSequence{\({\mathbb G},i\)}}{
		Initialize \(G' \leftarrow G_i\), \(w \leftarrow 1\),
                \(\textrm{total}_k \leftarrow {1}_{k\geqslant i}\), \(j \leftarrow
                1\), \(i_{\textrm{max}}\leftarrow n_e\)\;
		Remove all vertices with loops (hyperedges with only one vertex) in $G'$\;
		\While{\(V(G') \neq \emptyset\)}{
			Set $w \leftarrow w \times |V(G')|/j$ \;
			Set $j \leftarrow j+1$\;
			Select a random vertex $v \in V(G')$ and remove $v$ from $V(G')$ and from all elements of $E(G')$ that contain $v$; \;
			Remove all vertices with loops (hyperedges
                        with only one vertex) in $G'$\;
                        \For{\(k \leftarrow i+1\); 
                          \(k\leqslant i_{\textrm{max}}\)}{
			  \If{$v$ makes the set not independent for
                            $G_k$}{
                            $i_{\textrm{max}}\leftarrow k-1$\;
                            }
                        }                        
                        \For{\(k \leftarrow i\)
                          \KwTo\(i_{\textrm{max}}\)}{
			  Set $\textrm{total}_{k}\leftarrow \textrm{total}_{k}+w$ \;
                        }
		}
		\Return \(\textbf{total}\)\;
	}	
	\caption{Counting Independent Sets for a Nested
          Sequence of Hypergraphs}\label{CountISModified}
\end{algorithm}

Even the modified Algorithm~\ref{CountISModified} remains
computationally expensive. Since the sizing function only needs to be
accurately estimated near to the final cut-off, we develop
Algorithm~\ref{alg_find_cutoff}, which focuses on cut-off values close
to the chosen final cut-off. Starting from a collection of rough
estimates obtained using Algorithm~\ref{CountIS} with $N=1$ on a batch
of cut-off values, we apply isotonic regression to estimate the number
of independent sets for each cut-off. We use these estimates to select
a threshold cut-off, and concentrate the next batch of samples near to
this threshold. After each batch, we perform isotonic regression on
all samples collected so far to estimate the number of independent
sets for each cut-off, and use the results to update our threshold
estimate, which is the centre for the next batch of cut-off
values. Algorithm~\ref{alg_find_cutoff} is implemented in the
\texttt{HVS} package.

\setlength{\algomargin}{0pt} 
\setcounter{AlgoLine}{0}
\begin{algorithm}[htbp]
  \DontPrintSemicolon
  \SetAlgoNlRelativeSize{0} 
  \SetAlgoNlRelativeSize{0}
  \SetAlgoNlRelativeSize{0} \setlength{\textwidth}{500pt} 
  \renewcommand{\baselinestretch}{1.2} \SetAlgoLined
  \KwIn{Hypergraph \(G = (V, E)\), $p$-values $p_i$ for each
    edge, number of cut-offs per batch \(N\), threshold slope
    $\alpha$, number of batches $B$} \KwOut{Cut-off $c$ such
    that $|V|-\log_2(IS(G_c))=\alpha c$ }
  
  \SetKwFunction{SampleIndependentSequence}{SampleIndependentSequence}
  \SetKwFunction{IsotonicRegression}{IsotonicRegression}
  \SetKwFunction{ConvexHull}{ConvexHull}
  
  Select $c_1\leqslant\cdots\leqslant c_{N}$,
  evenly spaced between 0 and $q$ \tcp*[r]{solution must have
    $c\leqslant q$}
  \For{\(i \leftarrow 1\) \KwTo \(N\)}{
    \(w_i \leftarrow \SampleIndependentSequence(G_{c_i})\)
    \tcp*[r]{Function from Algorithm~\ref{CountIS}}
  }

  $\mathbf{s}\leftarrow\IsotonicRegression(\mathbf{w})$\;
  $\sigma_i\leftarrow |V|-\log_2(s_i)$ \tcp*[r]{Sizing function} 
  Find largest solution $i_0$ to $\sigma_i=\alpha c_i$
  \tcp*[r]{$\sigma$ a step function, $c_i$ linearly
    interpolated, so $i_0$ not necessarily an integer} 
  \For{$b\leftarrow 1$ \KwTo $B$}{
    Select $N$ cut-off values
    $\mathbf{c}_b=c_{b,1},\ldots,c_{b,N}$ where $c_{b,1}\leqslant\cdots\leqslant
    c_{b,N}$ centred around $c_{i_{b-1}}$\;
    \tcp*[r]{These cut-off values should be from the set of
      $p$-values. Add some randomness to the choice of cut-off
      values, so that the same values do not get repeated.}
    \For{\(i \leftarrow 1\) \KwTo \(N\)}{
      \(w_{b,i} \leftarrow
      \SampleIndependentSequence(G_{c_{b,i}})\) \tcp*[r]{Fn from
        Algorithm~\ref{CountIS}}
    }
    Merge $\mathbf{c}_b,\mathbf{w}_b$ into
    $(\mathbf{c},\mathbf{w})$, so that the merged $\mathbf{c}$ is in
    increasing order.
    $\mathbf{s}\leftarrow\IsotonicRegression(\mathbf{w})$\tcp*[r]{All samples from batches 1 to $b$}
    $\sigma_i\leftarrow |V|-\log_2(s_i)$ \tcp*[r]{Sizing function} 
    Find largest solution $i_b$ to $\sigma_i=\alpha c_i$
    \tcp*[r]{Not necessarily an integer}          
  }
  
  \SetKwProg{Fn}{Function}{:}{}
  \Fn{\IsotonicRegression{\(\mathbf{w}\)}}{
    $\kappa_i\leftarrow \sum_{j=1}^i w_j$\tcp*[r]{Cumulative Sum}
    $S\leftarrow\ConvexHull(\mathbf{i},\boldsymbol{\kappa})$\;
    $s_i\leftarrow\frac{\kappa_{S^+[i]}-\kappa_{S^{-}[i]}}{S^+[i]-S^{-}[i]}$\;
  }
  \Fn{\ConvexHull{\(\mathbf{x},\mathbf{y}\)}}{
    $m\leftarrow\texttt{length}(\mathbf{x})/2$\;
    $H_1\leftarrow\ConvexHull(\{x_i|i\leqslant m\},\{y_i|i\leqslant
    m\})$\;
    $H_2\leftarrow\ConvexHull(\{x_i|i\geqslant m\},\{y_i|i\geqslant
    m\})$\;
    Concatenate $H_1$ and $H_2$.\;
    Remove all non-convex points from the middle.\; 
  }
  \caption{Choosing the cut-off for GLSUP.}\label{alg_find_cutoff}
\end{algorithm}

To confirm that this improved computational efficiency does not
significantly affect performance, we compare
Algorithms~\ref{CountISModified} and~\ref{alg_find_cutoff} on several
simulation scenarios from Section~\ref{hvssim}. The results are in
Table~\ref{SUPPCompareAlgorithms} in Supplementary
Appendix~\ref{SUPPAppendixAlgorithms}.  There is no significant
difference between the algorithms in terms of any of the performance
metrics, but that Algorithm~\ref{alg_find_cutoff} is much faster.

\section{FDR Control for Hypergraph Variable Selection}\label{hvsfdr}

In this section, we set
$\mathcal{S} \subseteq \mathcal{P}(X)$ as the collection of all subsets 
of cardinality at most~$2$. Sets of size greater than~$2$ are rarely 
selected in practice, so including them yields minimal gains in power. 
Under this construction, for $|X|=p$ predictor variables, we have
$m=|\mathcal{S}|
= {p \choose 2} + p$.

\citet{SHRED} derive cut-offs for the GLSUP that control the gFDR under 
two regimes: (i) without any dependency assumptions on the $p$-values, 
analogous to the classical BY procedure, and (ii) under the assumption 
that the $p$-values satisfy the PRDS condition on the true nulls, as in 
the BH procedure~\citep{BY}. We recall the relevant results from 
\citet{SHRED} below. The GLSUP uses the following input:

\begin{enumerate}
	
\item A set $H_1,\ldots,H_m$ of null hypotheses with
  corresponding conservative $p$-values $P_1,\ldots,P_m$, which are
  random variables.
	
\item A given collection $\mathcal{C}\subseteq{\mathcal P}$ of
  ``rejectable'' subsets of $\{1,\ldots,m\}$, that is, sets of
  hypotheses that can be consistently rejected.  We assume ${\mathcal
    C}$ is closed under arbitrary intersections. This means that there
  is a corresponding closure operator that sends any $J\subseteq
  \{1,\ldots,m\}$ to the smallest set $\overline{J}\in{\mathcal C}$
  that contains $J$.
	
\item A sizing function $\sigma:{\mathcal
  P}\longrightarrow{\mathbb R}_{\geqslant 0}$, which is (not
  necessarily strictly) increasing with respect to set inclusion.
  
\end{enumerate}

With this setup,  \citet{SHRED} prove the following theorems:

\begin{theorem}[\citet{SHRED}, Theorem~1]\label{GeneralizedBY}
	Suppose that for any
	$J\subseteq\{1,\ldots,m\}$, $\sigma\left(\overline{J}\right)\leqslant \sum_{i\in
		J}w_i$ for some weights $w_1,\ldots,w_m$. Then the generalized
	linear step-up procedure with cut-off
	$$\alpha=\frac{\left(\sum_{i\in I_N}w_i\right)\left(1+\log(\sigma(\{1,\ldots,m\}))\right)-\sum_{i\in
			I_N}w_i\log(\sigma_i)}{q}$$ where
	$\sigma_i=\sigma(\{i\})$, controls the gFDR at level $q$.
\end{theorem}	

\begin{definition}[\citet{BY}]
	A random vector $X = (X_1,\ldots X_m) \in [0,1]^m$ is PRDS on $J
	\subseteq \{1,\ldots,m\}$ if for every upward-closed set $D
	\subseteq [0,1]^{m-1}$, $\forall i \in J$, $\pr(X_{\hat{i}} \in
	D|X_i = x)$ is non-decreasing in $x \in [0,1]$, where $X_{\hat{i}}$ denotes the vector $X$ with its $i$th element removed.
\end{definition}

\begin{theorem}[\citet{SHRED}, Theorem~2]\label{PRDS}
	Suppose that the conservative $p$-values $P_1,\ldots,P_m$ satisfy
	the PRDS condition on the set $I_N$ of true null hypotheses, and
	that the sizing function $\sigma:{\mathcal
		P}\longrightarrow{}{\mathbb R}_{\geqslant 0}$ is bounded by a sum
	of weights: $(\forall
	J\subseteq\{1,\ldots,m\})\left(\sigma\left(\overline{J}\right)\leqslant
	\sum_{i\in J}w_i\right)$ for some weights $w_i\geqslant 0$. Then the
	generalized linear step-up method with cut-off
	$\alpha=\frac{\sum_{i=1}^mw_i}{q}$ controls the gFDR at level $q$.
\end{theorem}

The key assumption in the preceding two theorems is that the sizing
function $\sigma : \mathcal{P} \to \mathbb{R}_{\geqslant 0}$ satisfies
an additive bound of the form
\[
(\forall\, J \subseteq \{1,\ldots,m\}) \ \left(
\sigma\!\left(\overline{J}\right) \leqslant \sum_{i \in J} w_i\right)
\]
for some nonnegative weights $w_i \geqslant 0$. Next we show that the
sizing function introduced in Section~\ref{hvsvarselection} satisfies
this additive bound condition.

\begin{lemma}[Additive bound for $\sigma$]\label{boundsigma}
Let $G=(X,S)$ be a hypergraph and define
\[
\sigma(S)=|X|-\log_2|IS(G)|.
\]
For any $S$, we have 
\[
\sigma(S)\le \sum_{i\in S}\sigma_i,
\quad\text{where }\sigma_i = k - \log_2(2^k-1)\text{ for an edge $i$ of size $k$}.
\]
\end{lemma}

\begin{proof}
We use induction on $|S|$. The case $S=\emptyset$ is immediate since 
$\sigma(\emptyset)=0$.

Let $S'=S\setminus\{i\}$ and let edge $i$ contain $k\ge1$ vertices.  
By definition, an independent set for $G$ is an independent set of
$G'=(X,S')$, i.e. 
\[
IS(G)\subseteq IS(G').
\]
We partition $IS(G')$ into $2^k$ disjoint classes according to the
intersection of each independent set with edge $i$. Let these classes
be denoted $IS(G')_A=\{U\in IS(G')|i\setminus U=A\}$. Exactly one class,
$IS(G')_\emptyset$ corresponds to independent sets that contain all
vertices in edge $i$. Since an independent set of $G'$ that does not
contain all the vertices of edge $i$ is an independent set of $G$, we
have $IS(G)=IS(G')\setminus IS(G')_\emptyset$. Now for any $U\in
IS(G')_\emptyset$, and any $A\subseteq i$, we have $U\setminus A\in
IS(G')_A$. Thus, $|IS(G')_\emptyset|\leqslant |IS(G')_A|$, and
therefore
\[
|IS(G)| \geqslant \frac{2^k-1}{2^k}\,|IS(G')|
\]

Thus
\[
\begin{aligned}
\sigma(S)
&= |X|-\log_2|IS(G)| \\
&\le |X|-\log_2\!\left(\tfrac{2^k-1}{2^k}|IS(G')|\right) \\
&= \sigma(S') + \log_2\!\left(\tfrac{2^k}{2^k-1}\right)
= \sigma(S') + \sigma_i.
\end{aligned}
\]
By induction, $\sigma(S')\le\sum_{j\in S'}\sigma_j$, so
\[
\sigma(S)\le \sum_{j\in S}\sigma_j.
\]
\end{proof}

For the specific case of the HVS procedure with $\mathcal{S} = \{S_i \subset
\mathcal{P}(X): |S_i| \leqslant 2\}$, Theorem~\ref{GeneralizedBY}
and Lemma~\ref{boundsigma} yield the following result.

\begin{corollary} \label{HVSarb}
	Let $\mathcal{S} = \{S_i \subset \mathcal{P}(X) : |S_i| \le 2\}$ be the collection of all singletons and all pairs of variables, and define the sizing function by $\sigma(S) =
	|X|-\log_{2}(|IS(G)|)$, where $IS(G)$ denotes the set of
        independent sets of the hypergraph $G = (X, S)$. Then the GLSUP with
	\begin{equation*}
		\alpha_{HVS} = \frac{1}{q}\left[p_0(1 + \log(p))+{p_0 \choose 2}\frac{\log(4)-\log(3)}{\log(2)}\left(1+\log(p) - \log\left(\frac{\log(4) - \log(3)}{\log(2)} \right) \right)\right]
	\end{equation*}
	where $p_0=|X_N|$ is the number of noise variables and $p=|X|$, controls the gFDR at level $q$.
\end{corollary}

This follows from the observation that under the hypergraph-based
hypothesis structure defined by $\mathcal{S} \subseteq \mathcal{P}(X)$
with subsets of cardinality at most $2$, we have
\[
\sigma_i =
\begin{cases}
1, & \text{if $i$ is a loop}, \\[4pt]
\dfrac{\log(4)-\log(3)}{\log(2)}, & \text{otherwise}.
\end{cases}
\]

We also know that $\sigma(\{1,\ldots,m\}) = p$. The index set $I_N$ consists of all loops at true nulls along with all edges between them. Thus $I_N$ contains $p_0$ hypotheses with $\sigma_i = 1$ and ${p_0 \choose 2}$ hypotheses with $\sigma_i = \frac{\log(4)-\log(3)}{\log(2)}$.

Although the PRDS condition is difficult to verify for most test statistics, it is commonly assumed in practice. Under PRDS, Theorem~\ref{PRDS} provides a smaller cutoff $\alpha$.

\begin{corollary}\label{HVSprds}
	Let $\mathcal{S} = \{S_i \subset \mathcal{P}(X) : |S_i| \le 2\}$ be the collection of all singletons and all pairs of variables, with the sizing function given by $\sigma(S) =
	|X|-\log_{2}(|IS(G)|)$.
 If the $p$-values satisfy the PRDS condition on $I_N$, the
	GLSUP with,
	\begin{equation*}
		\alpha_{HVS_{PRDS}} = \frac{1}{q}\left(p_0 + {p_0 \choose 2} \frac{\log(4) - \log(3)}{\log(2)}\right)
	\end{equation*}
	where $p_0$ is the number of noise variables, controls the gFDR at level $q$.
\end{corollary}   

The proof is the same as Corollary~\ref{HVSarb}. 

\section{Simulations}\label{hvssim}

\subsection{Simulation Design}

To assess the performance of the HVS method compared to other methods
designed for FDR control, simulations are conducted for Gaussian
regression, logistic regression, and Poisson regression. We
simulate data of the form $X \sim N(\mu,\Sigma)$, $\mu = (0,...,0)^T$,
$ \textrm{diag}(\Sigma) = 1$, under four different correlation
structures:

\begin{enumerate}[(1)]

\item Common component correlation:
$$\Sigma_{ij}=\left\{\begin{array}{ll}
1 & \textrm{if }i=j\\
\rho & \textrm{otherwise}
\end{array}\right.$$

\item Clustered correlation: randomly partition the variables into
  disjoint clusters $C_1,\ldots,C_K$ of sizes
  $|C_1|=5,|C_2|=10,|C_3|=15,\ldots$, with $C_K$ containing any
  leftover variables, so when $p=300$, $K=11$ and $|C_K|=25$, while
  when $p=200$, $K=9$ and $|C_K|=20$.
  $$\Sigma_{ij}=\left\{\begin{array}{ll}
1 & \textrm{if }i=j\\
\rho_{k} & \textrm{if }i\ne j\in C_k\\
0 & \textrm{otherwise}
\end{array}\right.$$
where each $\rho_k$ is independantly drawn from a uniform distribution
with a certain range. 

\item Autoregressive (AR(1))
  correlation:

  $$\Sigma_{ij}=\rho^{|i-j|}$$

\item Spatial correlation: 
Spatial correlation is imposed on a $100 \times 10$ grid of coordinates
$v_1, \ldots, v_{1000}$. The correlation between variables is defined as
\[
\cor(X_i, X_j)
= \rho \, e^{-\phi \lVert v_i - v_j \rVert},
\]
where $\lVert v_i - v_j \rVert$ denotes the Euclidean distance between grid points
$i$ and $j$. We set $\rho = 0.6$ and consider two decay parameters,
$\phi = 0.1$ and $\phi = 0.5$. The choice $\phi = 0.1$ produces a more
dispersed spatial field, with correlations remaining strong over larger
distances, while $\phi = 0.5$ corresponds to a more localized field in
which correlations decay rapidly and distant points exhibit very weak
correlation.

\end{enumerate}

For each setting, 100 simulation replicates are used for Gaussian and
Poisson regression, and 200 for logistic regression due to its
tendancy to produce wider confidence intervals.

Covariance matrices are generated using the \texttt{simstudy v0.8.1 R}
package~\citep{simstudy}. A new correlation matrix of the described
structure (common component, clustered,  AR(1) or Spatial correlation) was generated within
each replicate for each scenario. In all simulations $T = 100$ true
predictors were chosen at random. The coefficients $\beta_i$ are
simulated as i.i.d standard normal distributions, for $i=1,\ldots,
T$. Dataset sizes are fixed at $p = 300, n = 1000$ for Gaussian
regression, and $p = 200, n = 5000$ for logistic and Poisson
regression. The intercept term for the Logistic simulation is set so
that the marginal probability is 0.5 for each dataset, and Poisson
counts are generated with intercept fixed at -1, typically yielding
responses between 0 and 10.
With the
exception of the spatial case, these simulations match those used in
\citep{SHRED}, so the results for the shared methods (mirror method,
knockoff method, BH, BY and LASSO) are directly comparable.

Mirror methods follow~\cite{mirrors}: DS uses a single data split,
while MDS uses 50 data splits. Both DS and MDS results are obtained
following the \texttt{R} code given in \cite{mirror_code}. Knockoff
results are generated using the \texttt{knockoff v0.3.6 R}
package~\citep{knockoff_package}.  SHRED methods follow the procedures
for the SHRED, SHRED PRDS and SHREDDER methods outlined in
\citep{SHRED} with the hypothesis weight set at $w_S =
\log_2\left(\frac{2^k}{2^k - 1}\right)$, where $k = |S|$.

For each simulation and each replicate, we simulated a test set, of the
same size as the training data, for evaluating MSE and classification
accuracy. For Poisson regression, MSE is computed under the same
approach as the Gaussian case. All methods target gFDR control at $q =
0.05$.

The HVS results presented were computed using the
\texttt{hypergraph.sizing} and \texttt{GLSUP} packages. The \texttt{R}
code for the simulations is provided in the github repository
\url{https://github.com/s-organ/HVS-R}.

\subsection{Simulation Results}

Figures~\ref{normalHVS}--\ref{poissonHVS} present the gPower, gFDR,
and either MSE or classification accuracy for each method across all
simulation scenarios. For every method, mean squared error (MSE) and
classification accuracy were computed by first refitting a model using
only the variables selected by that method on the training data, and
then assessing predictive performance on an independent test dataset.

In each simulation, we include the LASSO’s MSE or classification
accuracy results in the figures to facilitate comparison of predictive
performance. Results for the LASSO’s power and FDR, however, appear
separately in Appendix~\ref{SUPPHVSsuppsims}. LASSO selections were
obtained using the penalty parameter that minimized the 10-fold
cross-validation error, and predictive performance was evaluated using
the resulting LASSO coefficients.

For all methods, gPower and gFDR are computed using the sizing
function $\sigma$ defined in this paper. For methods that select only
singleton variables, these correspond to the standard definitions of
FDR and power. In the logistic and Poisson simulations, the knockoff
method exhibited substantially lower power than competing approaches
and is therefore omitted from the corresponding figures. For the same
reason, the mirror method with DS is excluded from several plots. Full
results for these methods, together with detailed outcomes for each
simulation scenario, are shown in the tables in
Appendix~\ref{SUPPHVSsuppsims}.

To estimate MSE or classification accuracy for the SHRED methods, we
randomly select a single variable from each chosen set and use this
resulting set for prediction. For the HVS method, after identifying
all loops, we apply a greedy vertex-cover algorithm --- implemented
via the \texttt{R} package \texttt{gor}~\citep{gor} --- to determine a
representative set of predictor variables. The algorithm repeatedly
selects the vertex of highest degree, removes all edges containing
that vertex, and continues until all edges are covered. Any
vertex-cover algorithm would be expected to yield comparable results.


\begin{figure}[htbp]
	\centering
	\includegraphics[width=15cm, height = 15cm]{"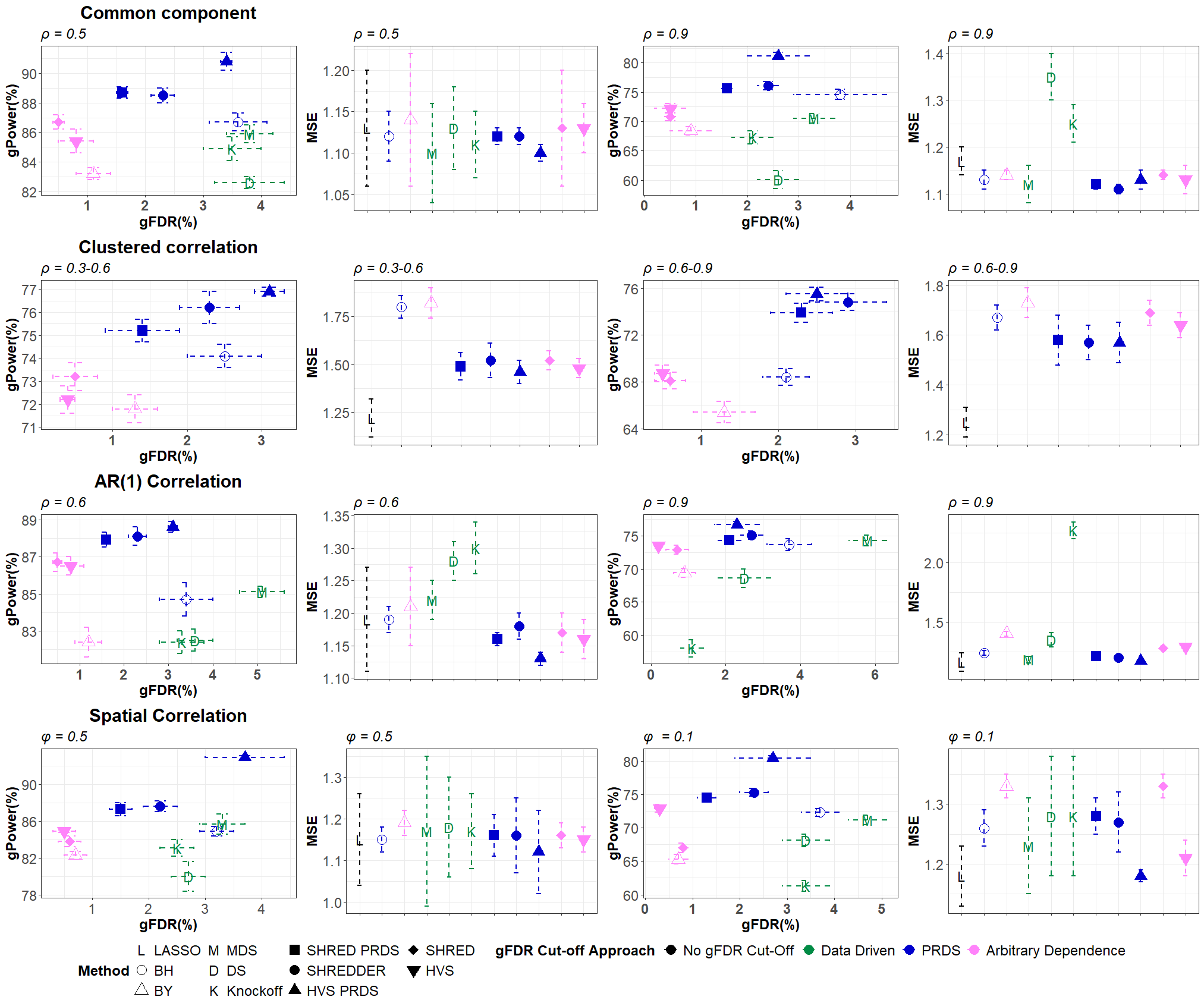"}
	\caption{Gaussian regression simulation results}
	\label{normalHVS}
\end{figure}

For the Gaussian simulations (Figure~\ref{normalHVS}), the HVS-PRDS
method consistently achieves higher gPower than the other methods that
rely on the PRDS assumption. Although this method tends to produce
slightly higher gFDR than the SHRED variants, the gFDR remains
controlled below the target level in all scenarios. Likewise, among
the methods with no dependence assumptions, the HVS procedure
consistently outperforms both BY and SHRED.

The primary advantage of HVS lies in its ability to select
intersecting sets of variables, enabling it to capture signal patterns
that are spatially or temporally structured --- such as under spatial
correlation or AR(1) dependence --- where true signals are dispersed
across correlated groups. In contrast, SHRED is restricted to disjoint
sets, which limits its effectiveness in these settings and results in
lower gPower.

Compared with data-driven approaches such as the mirror method and
knockoffs --- both of which exhibit decreasing power as correlation
strength grows --- the HVS methods maintain substantially higher
gPower along with lower MSE.

Examining MSE across methods provides insight into their practical
variable selection performance under FDR control. Notably, in the
spatial and AR(1) scenarios, the predictive accuracy of HVS-PRDS is
comparable to that of the LASSO, which does not control gFDR and is
specifically designed for achieving strong model prediction. Overall,
both variants of HVS (PRDS and arbitrary dependence) offer a favorable
trade-off between gPower, gFDR control and predictive performance,
particularly in complex correlation structures such as spatial and
autoregressive settings.

\begin{figure}[htbp]
	\centering
	\includegraphics[width=15cm, height = 15cm]{"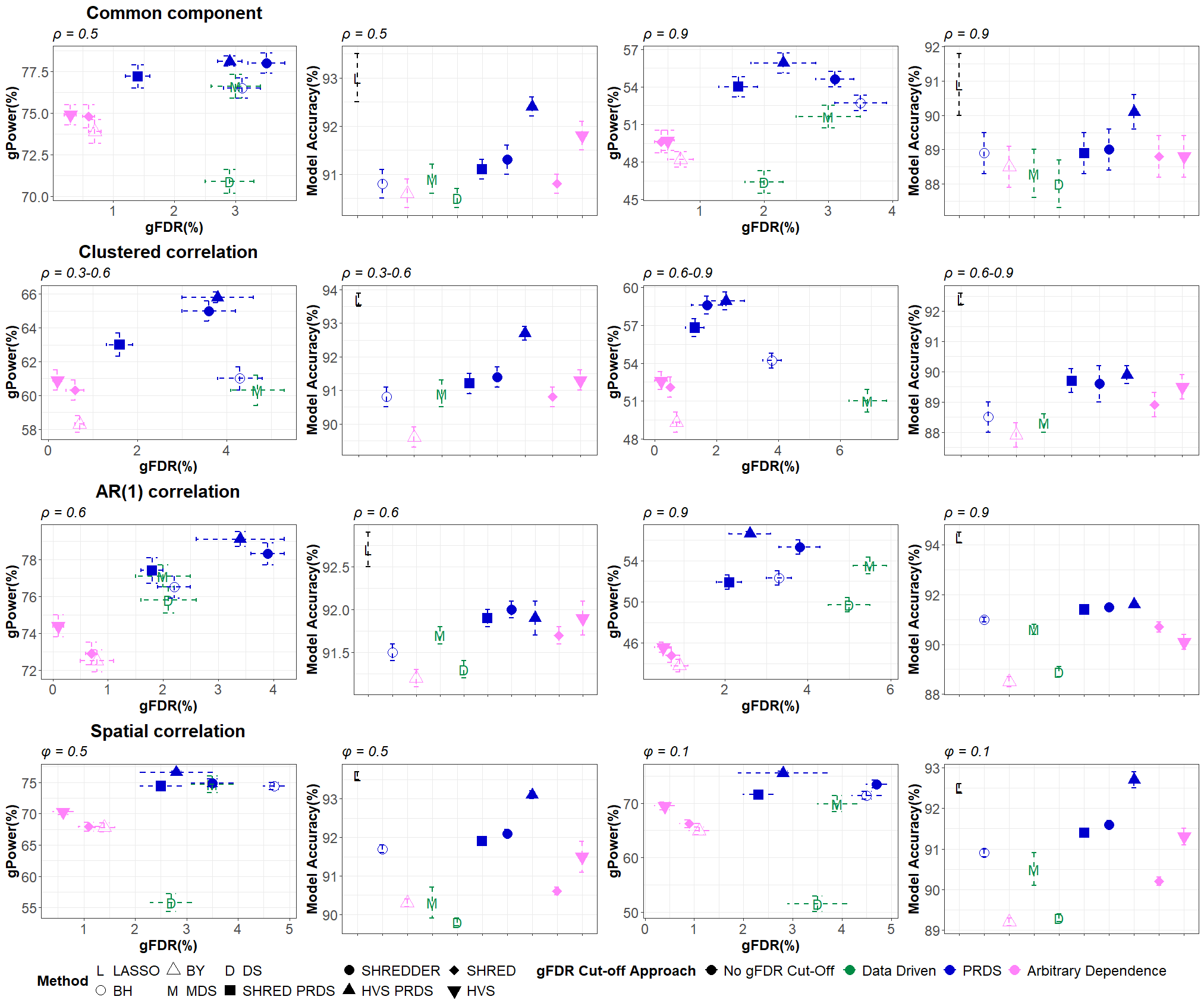"}
	\caption{Logistic regression simulation results}
	\label{logisticHVS}
\end{figure}

For the logistic regression simulations (Figure~\ref{logisticHVS}),
all methods exhibited reduced gPower relative to the Gaussian case, as
expected given the lower information content in the binomial
response. Nonetheless, the qualitative patterns observed in Gaussian
regression were largely retained. The HVS-PRDS method consistently
achieved higher gPower than other procedures assuming PRDS, while
controlling gFDR and maintaining strong predictive accuracy.
Comparable advantages were observed when comparing HVS with the BY and
SHRED methods, which also control gFDR with no dependence assumptions.

Among the data-driven procedures, only the mirror with MDS
method ever performed competitively, and only under low correlation
scenarios.

\begin{figure}[htbp]
	\centering
	\includegraphics[width=15cm, height = 15cm]{"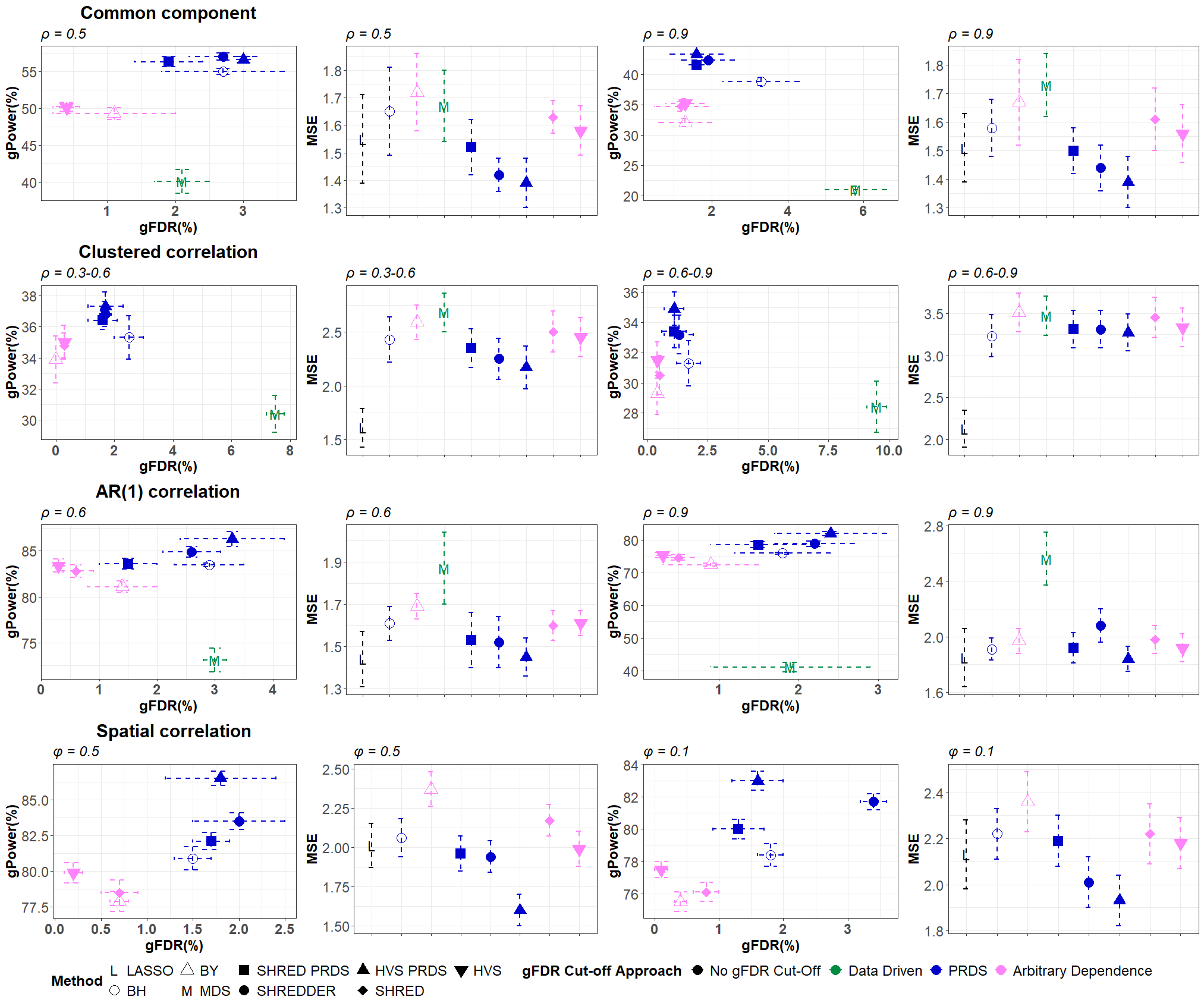"}
	\caption{Poisson regression simulation results}
	\label{poissonHVS}
\end{figure}

For the Poisson regression simulations (Figure~\ref{poissonHVS}), all
methods showed lower gPower compared to the Gaussian case. Overall
trends seen in the Gaussian regression and Logistic regression cases are
also seen here. The HVS-PRDS tended to have a higher gPower than
other PRDS methods while maintaining gFDR control. The improvement in
this method over others was again most evident in the spatial and
AR(1) correlation settings. The data-driven methods performed badly
here, neither achieving high power nor controlling FDR. Full results
are provided in Appendix~\ref{SUPPHVSsuppsims}. 

\subsection{Simulation on Real Microbiome X}

We evaluated the performance of the HVS method relative to other
FDR-controlling procedures using data with real-world correlation
structure. For this analysis, we used the MIDAS 4 global wastewater
microbiome taxonomy dataset~\citep{midas}, which provides whole-number
read counts for taxa collected from wastewater treatment plants
worldwide. Following standard recommendations for adjusting
for sequencing depth in microbiome
studies~\citep{quinn2019compositions,yerke2024proportion} we first
added a small constant ($10^{-6}$) to each count, then converted
to relative abundances and performed a centred log-ratio
transformation. After removing rare and conditionally abundant taxa,
the resulting dataset contained $n = 1278$ samples and $p = 604$
variables.

Using the CLR-transformed relative abundances of each genus as predictors,
we generated 100 simulated datasets with a normally distributed response
$Y \sim N(X\beta, 1)$. For each simulation, we randomly selected
$T = 200$ true nonzero coefficients, with each nonzero $\beta$
independently drawn from a $N(0,1)$ distribution. Variable selection was
performed on the full dataset, and predictive performance was evaluated
using 10-fold cross-validated MSE based on the selected variables. All
methods were assessed at the target FDR level $q = 0.05$.

As in the simulations shown in
Figures~\ref{normalHVS}--\ref{poissonHVS}, the simulated $Y$ and
$\beta$ coefficients are identical to those used in \citep{SHRED}, so
methods shared between the two studies yield identical results.
However, because the SHRED procedures here are weighted according to
the HVS scheme, their outcomes differ slightly from those originally
reported in \citep{SHRED}.

\begin{table}[htbp]
  \centering
  \caption{\label{midassim} Microbiome $X$ data with simulated Gaussian $Y$}
  \small
  \begin{tabular}{cccccccc}
    \hline
	{Method} & {gFDR Cut-off} & {gPower} & {SE gPower} & {gFDR} & {SE gFDR} & {MSE} & {SE MSE} \\
	\hline
	LASSO & No gFDR Cut-off & 88.6\% & 0.2\% & 46.3\% & 0.9\% & 1.17 & 0.04\\
	\hline
	Mirror & MDS (Data driven) & 70.2\% & 0.5\% & 6.7\% & 0.2\% & 1.53 & 0.06\\
	Mirror & DS (Data driven) & 66.5\% & 0.5\% & 3.3\% & 0.2\% & 2.28 & 0.10\\
	Knockoff & Data driven & 69.0\% & 0.8\% & 3.7\% & 0.3\% & 2.68 & 0.05\\
	\hline
	BH & PRDS &  82.3\% & 0.7\% & 3.6\% & 0.4\% & 1.33 & 0.03 \\
	SHRED & PRDS & 84.9\% & 0.2\% & 1.8\% & 0.4\% & 1.31 & 0.04\\ 
	SHREDDER & PPRDS & 85.2\% & 0.5\% & 3.8\% & 0.3\% & 1.28 & 0.04\\
	HVS & PRDS & 85.8\% & 0.7\% & 3.5\% & 0.4\% & 1.25 & 0.08 \\
	\hline
	SHRED & Arbitrary & 80.0\% & 0.2\% & 0.5\% & 0.1\% & 2.08 & 0.06\\
	BY & Arbitrary & 77.7\% & 0.8\% & 0.7\% & 0.3\% & 2.13 & 0.08\\
	HVS & Arbitrary & 82.4\% & 0.7\% & 0.6\% & 0.2\% & 2.04 & 0.09 \\
	\hline
  \end{tabular}
\end{table}

The methods using data‑driven gFDR cutoffs did not perform well in
this setting, highlighting the problems these methods may face in the
long-tailed predictor setting. In contrast, the setwise selection
methods (SHRED and HVS) achieved high power while controlling the gFDR
below the desired level, demonstrating their advantage over
traditional single variable selection methods in highly correlated
data. The HVS methods performed slightly better than the corresponding
SHRED methods in terms of generalized power and MSE, although the
differences were not significant.

\section{Case Study}\label{hvscasestudy}

We analyzed the WHO daily global COVID-19 case counts \citep{covid},
using the dataset available on April~25,~2025. To predict the number
of cases on a given day from the previous 39 days, we reorganized the
data so that each row contained 40 consecutive days of counts for a
country. The first 39 days (``Day~39'' to ``Day~1'') served as
predictors, and the 40th day (``Current Day'') was the response, with
``Day~1'' denoting the day immediately preceding the Current Day.

Any 40‑day sequence that contained a zero case count on any day
(including the current day) was excluded. After filtering, the dataset
contained 4,527 samples of 40‑day case-count sequences drawn from 153
countries, with each country contributing between 1 and 50 sequences
of consistently non‑zero new COVID‑19 cases. Because case counts
varied substantially across countries, all values were log‑transformed
before analysis.

Given the strong likelihood of within-country correlation, we used a
linear mixed‑effects model with a random intercept for country to
perform variable selection. Compared with a standard linear model, the
mixed-effects model produced residuals with reduced autocorrelation,
and yielded a lower AIC. Furthermore, an $F$-test indicated that the
random-effect term was significant. As this is not a generalized
linear model, results are not available for the mirror methods,
knockoff method, or LASSO.  The variables were selected on the entire
data and the MSE was computed for each method using a 10-fold cross
validation on the selected predictors. Table~\ref{covidresults}
summarizes the model selection of each available method. The size of
the rejection for each method is calculated using that method's
respective sizing function.

\begin{table}[htbp]
  \centering
  \caption{\label{covidresults} COVID-19 case study model selection results}
  \begin{tabular}{cclcc}
    \hline
    {Method} & {gFDR Cut-off} &  $\sigma(\mathcal{R})$ & {MSE} & {SE MSE}\\
    \hline
    BH & PRDS & 23 & 0.124 & 0.005 \\
    SHRED & PRDS & 23 & 0.124 & 0.005 \\
    SHREDDER & PPRDS & 22 & 0.125 & 0.005 \\
    HVS & PRDS & 24.9 & 0.114 & 0.005\\
    \hline
    SHRED & Arbitrary & 17 & 0.128 & 0.005 \\
    BY & Arbitrary & 17 & 0.128 & 0.005 \\
    HVS & Arbitrary & 19.3 & 0.117 &  0.004\\
    \hline
  \end{tabular}
\end{table}

\begin{figure}[htbp]
	\centering
	\medskip
	\subfigure[Selected variables (blue --- selected, red --- in the vertex cover)]{\includegraphics[width=8cm, height =
            4cm,clip=TRUE,trim=0 0 144 0]{"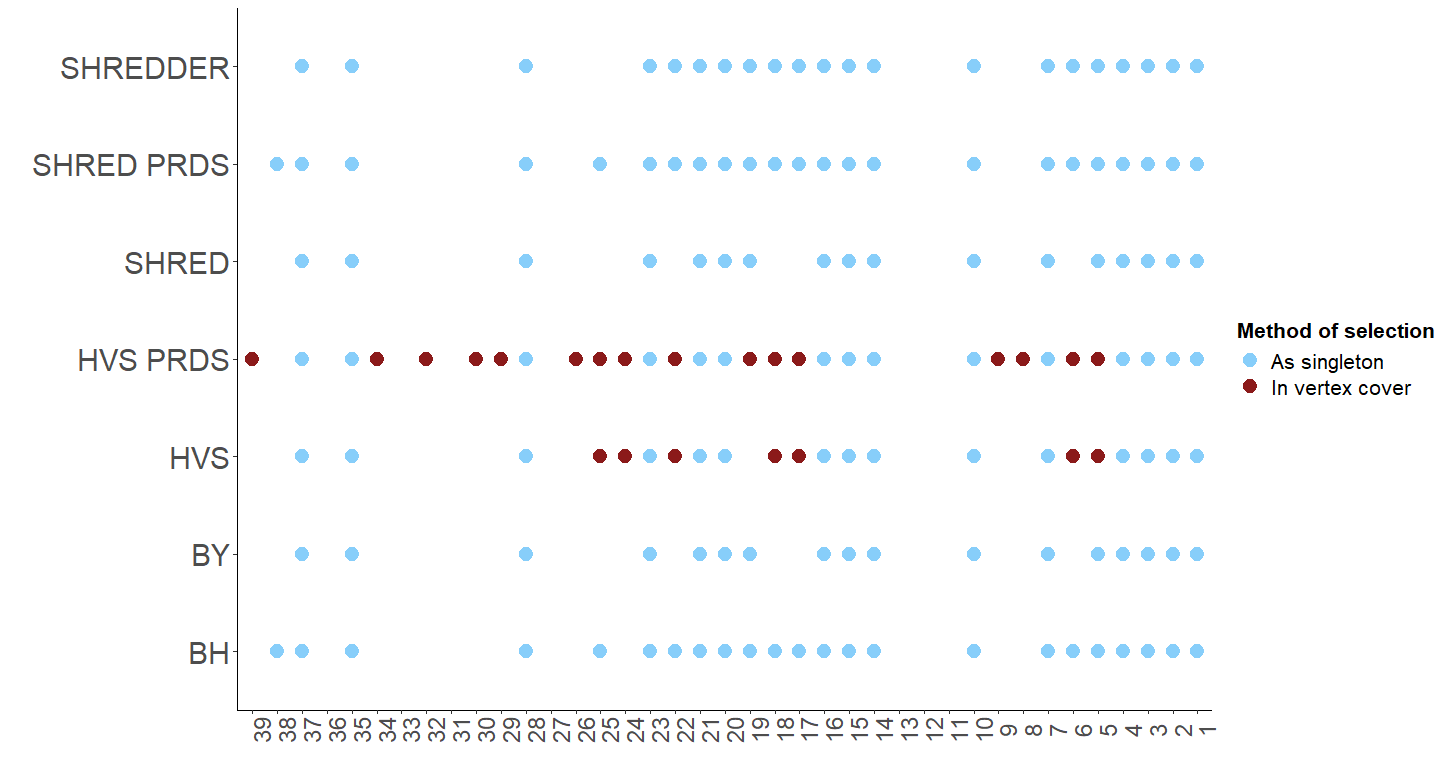"}}
	\subfigure[HVS selected graph]{\includegraphics[width=6cm, height = 4cm]{"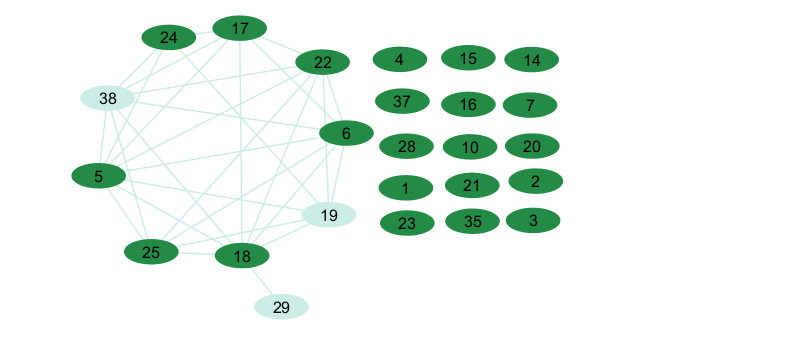"}}

	\subfigure[Sizing Function vs. $p$-value for
          HVS-PRDS]{\includegraphics[width=13cm,clip=TRUE,trim= 0 0 20 50]{"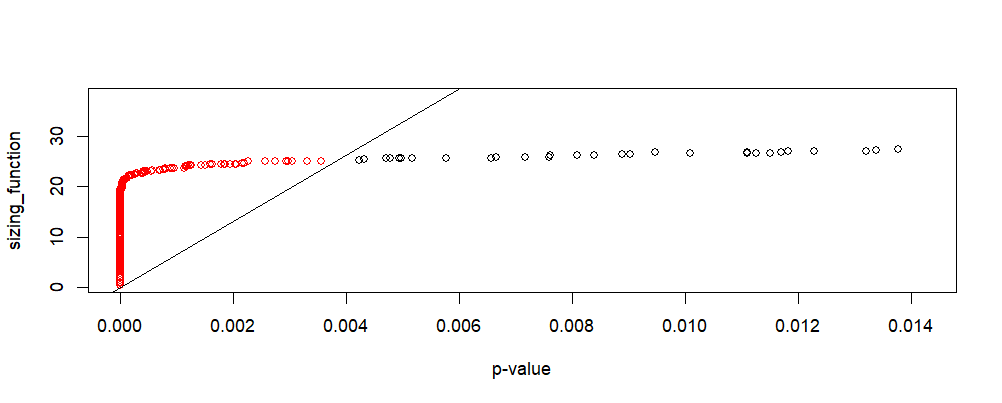"}}
	\caption{Covid case study results.}\label{covidsel}
\end{figure}

Figure~\ref{covidsel}(a) shows the dates selected by each method. For
clarity, days are labeled numerically (e.g., ``Day~1'' is shown as 1).
For the HVS and HVS-PRDS methods in Figure~\ref{covidsel}(a), we show
the variables included in one vertex cover of the selected graph. The
SHRED methods selected only singletons. Figure~\ref{covidsel}(b) shows
the full HVS selection graph, with the vertex cover
highlighted. (HVS-PRDS selected the same singletons but produced a
denser graph.)

We see that the variables selected by HVS and HVS-PRDS are able to
more accurately predict the log case counts of COVID than the
variables selected by other methods. SHRED, SHRED-PRDS and SHREDDER
are unable to select any families of predictors, and thus SHRED-PRDS
and SHRED give the same results as BH and BY. On the other hand HVS is
able to identify a number of pairs of predictors, where the individual
predictors cannot be reliably determined, allowing a much better
predictive model.  Comparing the variables selected by BY and HVS, we
see that BY has selected Days 5 and 19 that were not selected
individually by HVS. These days were included in a large cluster of
days selected, which also includes a number of days that were not
selected by BY. The days selected by BH form a vertex cover of the
hypergraph selected by HVS. It is not quite a minimal vertex cover, as
one of Days~6, 18, 22, or 25 can be removed. This indicates that even
using a strict cut-off for FDR control, HVS is able to identify a
large number of strong predictors. With the PRDS cut-off, HVS is able
to select many more sets of days, and the days selected by BH do not
form a vertex cover, indicating that BH has missed some true
predictors.

Figure~\ref{covidsel}(c) shows the sizing-function applied to the set
of hypotheses with $p$-values below a cut-off. All hypotheses are
shown on the plot, even hypotheses that imply hypotheses with smaller
$p$-values, and therefore do not increase the sizing function. The
sizing function was estimated using
Algorithm~\ref{CountISModified}, so there is some variability,
but we can see that the HVS-PRDS cut-off falls in a relatively
wide-gap between two $p$-values, indicating that this approximation
probably did not influence the rejected hypotheses. We also see that
there are a number of hypotheses with extremely small $p$-values, then
some hypotheses with less small $p$-values, which probably correspond
to false nulls, but less important variables. Many of these hypotheses
do not correspond to an increase in sizing function, most likely
indicating hypotheses that have already been implicitly rejected upon
rejecting a weaker hypothesis with smaller $p$-value. However, we do
still observe an increase in sizing function as the $p$-value cut-off
increases, indicating that these represent new discoveries.
 
\section{Conclusion}\label{hvsconclusion}

In this article, we extend the framework of~\citep{SHRED}, which
selects variables in sets when strong correlations make it impossible
to distinguish true predictors from surrogate variables. While the
SHRED procedures yield disjoint sets through a hierarchical clustering
structure, the HVS method applies generalized linear step‑up
procedures to possibly overlapping sets. This flexibility is often
advantageous: it enables the identification of better predictive
variables while maintaining control of the generalized FDR. A key
component of enabling variable selection with non‑disjoint sets was
defining an appropriate measure of the “number of variables selected.”
For this, we used a sizing function based on the number of independent
sets of the hypergraph, which naturally generalizes the usual count in
the single‑variable case.

Our simulation studies show that the HVS method can outperform
existing gFDR‑controlling methods when predictors are highly
correlated. Notably, although HVS is designed for gFDR control rather
than prediction, it often achieves MSE comparable to that of LASSO, a
method explicitly optimized for predictive accuracy. This highlights
the benefit of testing sets of surrogate variables, allowing more
effective rejection of hypotheses involving correlated predictors. The
advantages of HVS were especially pronounced under spatial correlation
structures, suggesting that the hypergraph-based selection approach
may be particularly well suited for these cases.

Here, we restricted attention to hypergraphs consisting of singletons
and pairs, but the sizing function defined in this article applies to
arbitrary hypergraphs. Extending the method to larger sets could be
valuable in extreme correlation settings where groups of three or more
predictors are indistinguishable. However, the number of sets grows
rapidly with set size, causing the cut-offs $\alpha$ from
Theorems~\ref{GeneralizedBY} and~\ref{PRDS} to increase and thereby
reducing gPower. Effective use of larger sets may therefore require
further methodological developments, potentially improving upon the
results of \citep{SHRED}.
 
\section*{Supplementary Material}

The supplementary material includes three appendices.
Appendix~\ref{SUPPCountIndepProof} contains a more detailed
explanation of the algorithm for approximately counting the number of
independent sets, along with an explanation of why it converges to the
correct answer, and a brief discussion of its computational
complexity. Appendix~\ref{SUPPHVSsuppsims} contains full simulation
results for all simulations, including some methods that were excluded
from the figures. Appendix~\ref{SUPPAppendixAlgorithms} compares
Algorithms~\ref{CountISModified} and~\ref{alg_find_cutoff} in terms of
accuracy and computation time in a number of simulation scenarios.

\section*{Declaration of the use of generative AI and AI-assisted technologies}

During the preparation of this work the authors used Copilot to assist
with writing some early drafts of the manuscript. After using this
tool, the authors reviewed and edited the content as necessary
and take full responsibility for the content of the publication.

\bibliography{reference.bib}
\end{document}